\documentclass[conference]{IEEEtran}
\IEEEoverridecommandlockouts
\usepackage{cite}
\usepackage{amsmath,amssymb,amsfonts}
\usepackage{amsthm}
\usepackage{algorithmic}
\usepackage{graphicx}
\usepackage{textcomp}
\usepackage{xcolor}
\usepackage{graphicx}
\usepackage[normalem]{ulem}

\usepackage{booktabs}
\usepackage{makecell}
\usepackage{multirow}

\usepackage{amsmath,amsfonts,bm}

\def\eqref#1{equation~\ref{#1}}

\def\1{\bm{1}}

\def\rmX{{\mathbf{X}}}
\def\rmY{{\mathbf{Y}}}
\def\rmZ{{\mathbf{Z}}}

\DeclareMathAlphabet{\mathsfit}{\encodingdefault}{\sfdefault}{m}{sl}
\SetMathAlphabet{\mathsfit}{bold}{\encodingdefault}{\sfdefault}{bx}{n}

\newcommand{\KL}{D_{\mathrm{KL}}}

\usepackage{algorithm}
\usepackage{algorithmic}
\usepackage{amsmath}
\usepackage{amssymb}
\usepackage{amsfonts}
\usepackage{color}

\usepackage{url}
\usepackage{graphicx}

\newtheorem{proposition}{Proposition}
\newcommand{\indep}{\perp\!\!\!\perp}
\usepackage{CJKutf8}
\usepackage{wrapfig}
\usepackage{colortbl}
\definecolor{mygray}{gray}{.8}
\definecolor{mylightgray}{gray}{.9}
\def\BibTeX{{\rm B\kern-.05em{\sc i\kern-.025em b}\kern-.08em
    T\kern-.1667em\lower.7ex\hbox{E}\kern-.125emX}}
\begin{document}

\title{Mind the Gap: Promoting Missing Modality Brain Tumor Segmentation with Alignment\\
\thanks{The work was partially supported by the following: National Natural Science Foundation of China under No.92370119, No. 62206225, and No. 62376113; 
Jiangsu Science and Technology Program (Natural Science Foundation of Jiangsu Province) under No. BE2020006-4;
Natural Science Foundation of the Jiangsu Higher Education Institutions of China under No. 22KJB520039;
XJTLU Research Development Funding 20-02-60. 
Computational resources used in this research are provided by the School of Robotics, XJTLU Entrepreneur College (Taicang), Xi'an Jiaotong-Liverpool University.}
}

\author{\IEEEauthorblockN{Tianyi Liu}
\IEEEauthorblockA{\textit{School of Robotics} \\
\textit{Xi’an Jiaotong-Liverpool University}\\
Suzhou, China\\
Tianyi.Liu2203@student.xjtlu.edu.cn}
\and
\IEEEauthorblockN{Zhaorui Tan}
\IEEEauthorblockA{\textit{School of Advanced Technology}\\
\textit{Xi'an Jiaotong-Liverpool University}\\
Suzhou, China \\
Zhaorui.Tan21@student.xjtlu.edu.cn}
\and
\IEEEauthorblockN{Haochuan Jiang}
\IEEEauthorblockA{\textit{School of Robotics} \\
\textit{Xi’an Jiaotong-Liverpool University}\\
Suzhou, China \\
h.jiang@xjtlu.edu.cn}

\and
    \IEEEauthorblockN{Xi Yang}
\IEEEauthorblockA{\textit{School of Advanced Technology} \\
\textit{Xi’an Jiaotong-Liverpool University}\\
Suzhou, China \\
xi.yang01@xjtlu.edu.cn}
\and
\IEEEauthorblockN{Kaizhu Huang}
\IEEEauthorblockA{\textit{Data Science Research Center}\\
\textit{Duke Kunshan University}\\
Suzhou, China \\
kaizhu.huang@dukekunshan.edu.cn	}
}
\maketitle

\begin{abstract}
Brain tumor segmentation is often based on multiple magnetic resonance imaging (MRI). However, in clinical practice, certain modalities of MRI may be missing, which presents an even more difficult scenario. To cope with this challenge, knowledge distillation has emerged as one promising strategy. However, recent efforts typically overlook the modality gaps and thus fail to learn invariant feature representations across different modalities. Such drawback consequently leads to limited performance for both teachers and students. To ameliorate these problems, in this paper, we propose a novel paradigm that aligns latent features of involved modalities to a well-defined distribution anchor. As a major contribution, we prove that our novel training paradigm ensures a tight evidence lower bound, thus theoretically certifying its effectiveness. Extensive experiments on different backbones validate that the proposed paradigm can enable invariant feature representations and produce a teacher with narrowed modality gaps. This further offers superior guidance for missing modality students, achieving an average improvement of $1.75$ on dice score.
\end{abstract}


\section{Introduction}
\label{sec:intro}
Malignant brain tumors severely threaten people's lives. Accurate brain tumor segmentation is crucial for treatment planning~\cite{chen2021learning}.
Multiple Magnetic Resonance Imaging (MRI), such as Fluid Attenuation Inversion Recovery (Flair), contrast-enhanced T1-weighted (T1ce), T1-weighted (T1) and T2-weighted (T2), are common tools to segment brain tumors~\cite{zhao2022modality}.
Since different modalities complement each other in understanding physical structure and physiopathology, combining them may naturally improve tumor segmentation~\cite{lindig2018evaluation,menze2014multimodal,patil2013medical,maier2017isles}. 
However, due to difficulties such as data corruption and scanning protocol variations, in real clinical practice, 
certain modalities may often be missing~\cite{chen2023query,qiu2023scratch,qiu2023modal,liu2021incomplete}.
Therefore, designing a generalized multi-modal approach to overcome difficulties brought by missing modalities is critical for practical clinical applications.

\begin{figure*}[t]
  \centering
  \includegraphics[width=\linewidth]{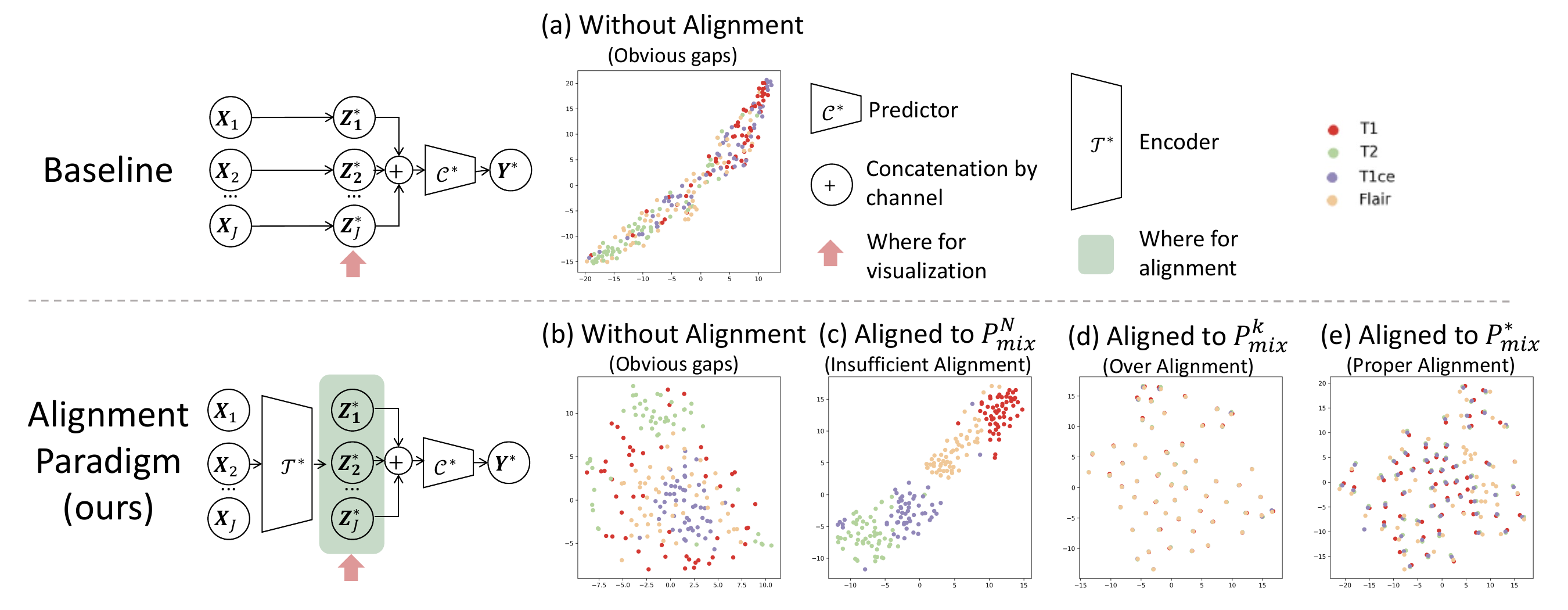}
    \caption{
    Comparison results of T-SNE maps for \emph{teachers}' latent features that are trained using different paradigms, whose structures are exhibited as diagrams.
    Data samples from each modality are assigned with a distinct color.
    For our paradigm, the latent space of each model is aligned to different empirical forms of $P_{mix}$: $P_{mix}^{N}$, $P_{mix}^{k}$ and $P_{mix}^*$ which will be discussed in Sec.~\ref{sec:align_details}.    
   }
   \label{fig:dist_vis}
\end{figure*}


To address this challenge, knowledge distillation (KD) has emerged as one promising solution. Recent efforts in KD~\cite{azad2022smu,chen2021learning,hu2020knowledge,wang2023prototype} initially train a \emph{teacher} with complete modalities that will then be used to supervise \emph{students} to access missing modalities.
To distill knowledge from \emph{teachers}, KD-Net~\cite{hu2020knowledge} employs the Kullback-Leibler (KL) loss to minimize the latent space divergence between \emph{teachers} and \emph{students}; 
PMKL~\cite{chen2021learning} is later designed to enhance KD-Net by incorporating contrastive loss.
Besides, ProtoKD~\cite{wang2023prototype} engages a prototype knowledge distillation loss to encourage simultaneous intra-class concentration and inter-class divergence. Moreover, Style matching U-Net addresses this problem by disentangling content and style components in the latent space~\cite{azad2022smu}.

While the above-mentioned wisdom enhances the segmentation capabilities of \emph{students}, their effectiveness is limited by their \emph{teachers} which remain sub-optimal and insufficiently explored.
Concretely, in these methods, \emph{teachers} simply treat different modalities as distinct channels and typically ignore the modality gaps. However, given that these MRI are captured by different imaging principles,  modality gaps unfortunately exist as always (see Figure~\ref{fig:dist_vis} (a)). As such, the \emph{teachers} may fail to learn invariant features, which further prevents the model from learning shared representations across different modalities and consequently degrades the prediction performance.

Illuminated from alignment approaches in narrowing domain gaps in classification tasks~\cite{ganin2016domain,hu2020domain,li2018deep}, in this paper, we investigate whether alignment can also be used to reduce modality gaps for brain tumor segmentation tasks.
To this end, we propose a novel \textbf{alignment paradigm} where~\emph{teachers} can indeed provide enhanced guidance to \emph{student}
by latent space alignment.
Specifically, latent features of modalities are initially placed in the same space by employing an encoder. 
Inspired by VAE~\cite{tomczak2018vae} and HeMIS~\cite{havaei2016hemis},  latent features of modalities are aligned to a pre-defined latent distribution as the anchor (termed as $P_{mix}$).
It is worth noting that based on a series of theoretical analyses, we validate several possible empirical forms for
$P_{mix}$. 
By analyzing the empirical evidence and the T-SNE~\cite{vandermaaten08a} visualizations displayed in Figure~\ref{fig:dist_vis}, 
we reveal that the best $P_{mix}$ is 
obtained by the weighted combination of each modality. 
As shown in Figures~\ref{fig:dist_vis} (c) and (d), other options of $P_{mix}$'s yield inferior results.
Furthermore, with the best $P_{mix}$, the proposed alignment paradigm fosters the learning of modality-invariant features by reducing modality gaps, producing a significant improvement on different modalities and different backbones.
The major contributions of the paper are summarized as follows:
\begin{itemize}
    \item {We invent a novel alignment paradigm including a latent space distribution $P_{mix}$ as the aligning anchor to learn cross-modality invariance.
    }
    \item We provide theoretical support for the proposed alignment, showing that individually aligning each modality to the best $P_{mix}$ certifies tighter Evidence Lower Bound than mapping all modalities as a whole to the $P_{mix}$.
    \item  With extensive experiments, we verify the superiority of the proposed paradigm in promoting the brain tumor segmentation performance of both the \emph{teacher} and the \emph{students} in the latest state-of-the-art backbones.
\end{itemize}

\section{Theoretical Motivations}


\textbf{Notations.} Considering $J$ modalities of medical images with paired observations and targets $\{\rmX_j\}_{j=1}^J$ and $\rmY$. 
Note for medical modalities,  $\rmY$ remains static for all modalities. 
For the \emph{teacher} of medical segmentation with missing modalities, the encoders are denoted as 
$\mathcal{T}: \mathcal{T}(\rmX_j) \to \rmZ_j^*$ where $\rmZ_j^*$ represents the produced latent features. 
Simultaneously, a predictor $\mathcal{C}$ that predicts segmentation masks from $\{\rmZ_j^*\}_{j=1}^J$ as $\mathcal{C}^*:\mathcal{C}^*(\{\rmZ_j^*\}_{j=1}^J) \to \rmY$. 
Correspondingly, we denote the possible downstream model for the $j^{th}$ target modality as $\mathcal{S}_j:  \mathcal{S}_j(\rmX_{j}) \to \rmZ_{j}$ of each modalities with their predictor $\mathcal{C}_{j}:\mathcal{C}(\rmZ_{j}) \to \rmY$.
Let $P(\cdot)$, $\KL(\cdot\Vert\cdot)$, $H_{c}(\cdot,\cdot)$, $I(\cdot;\cdot)$ denote the probability of a random variable from the distribution, 
KL divergence, cross-entropy, and mutual information, respectively.

\textbf{Previous methods.} 
The original objective used in ~\cite{hu2020knowledge,chen2021learning,wang2023prototype,azad2022smu} of training \emph{teacher}  can be treated as using a fixed linear 
$\mathcal{T}$. Thus its objective 
is:
\begin{align}
\label{eq:t_obj_ori}
    & \max_{\mathcal{C}^*} \!\sum\nolimits_{j=1}^J \!\mathbb{E}_{\rmZ_j^*\sim P(\rmZ_j^*)}[\ln P(\rmY \! \mid \! \mathcal{C}^*(\rmZ_j^*))].
\end{align}
In the scope of information theory, it can be altered as:
\begin{align}       
    &\min_{\mathcal{C}^*} \sum\nolimits_{j=1}^J H_c(P( \mathcal{C}^*(\rmZ_j^*)), P(\rmY)).  
\end{align} 
Meanwhile, for $\mathcal{S}_j$ that leverages knowledge from $\mathcal{T}$, its objective is:
\begin{equation}
\label{eq::ss_obj}
    \max_{\mathcal{S}_j,\mathcal{C}_j}
   \mathbb{E}_{\rmZ_j\sim P(\rmZ_j)}[\ln P(\rmY \mid \mathcal{C}(\rmZ_j))]  
   - \KL(P(\rmZ_j) \Vert P(\rmZ_j^*)).   
\end{equation}
In practice, the modality which $\mathcal{S}_j$ aims to is unknown for $\mathcal{T}$. Thus, we expect the sum of risks for all possible \emph{students} (shown in Eq.~(\ref{eq::ss_obj})) to be minimized:
\begin{align}
   \max_{\mathcal{S}_j,\mathcal{C}_j}  \sum_{j=1}^J [
    \mathbb{E}_{\rmZ_j\sim P(\rmZ_j)}&\![\ln\! P(\rmY \! \mid \!\mathcal{C}(\rmZ_j))] \!\!- \!\!\KL(P(\rmZ_j) \Vert P(\rmZ_j^*)) ] \nonumber \\
    =\min_{\{\mathcal{S}_j\}_{j=1}^J, 
    \{\mathcal{C}_j\}_{j=1}^J} &\sum\nolimits_{j=1}^J [\KL(P(\rmZ_j) \Vert P(\rmZ_j^*)) 
    \\
    & +
     H_c(P( \mathcal{C}_j(\rmZ_j)), P(\rmY))].
    \nonumber
\end{align}

Figure~\ref{fig:dist_vis}(a) illustrates that multi-modal medical images often exhibit incomplete space coverage and significant modality gaps, which can degrade the performance of \emph{teachers}. In contrast, our experiments validate that the aligned latent space produced by our approach improves the generalization ability of the medical \emph{teacher}, then benefiting downstream \emph{students}.

\textbf{Our alignment paradigm}
To alleviate the modality gaps,
our alignment paradigm aligns all modal latent features to a pre-defined distribution $P^{k}_{mix}$, 
as shown in Figure~\ref{fig:dist_vis}(b) with $P^{k}_{mix}$ and $P^*_{mix}$ column. 
This part provides more details about our approach.
Different from previous studies~\cite{azad2022smu,chen2021learning,hu2020knowledge,wang2023prototype},
we further define a continuous distribution $P_{mix}$ as the targeted latent space distribution of $\rmZ^*$
for the \emph{teacher}.

\begin{proposition}
\label{prop:P_mix}
For training a multi-modal \emph{teacher}
model, it is assumed that $\rmZ_i \indep \rmZ_j$ where $i,j \in \{1,..., J\}, i \neq j$. In this scenario, there exists a probability distribution $P_{mix}$ that can be used as an anchor distribution to align the latent variables $\rmZ^*$, while preserving sufficient information for accurate prediction of the segmentation labels $\rmY$.
\end{proposition}

\begin{proof}
The modality-independent assumption is derived from the fact that each modality is independent of each other. 
If $P_{mix}$ preserves sufficient information for accurate prediction of the segmentation labels $\rmY$,
based on the joint and marginal mutual information, we have 
\begin{align}
\label{eq:align}
    & \sum\nolimits_{j=1}^{J} \!I(P_{mix}(\rmZ_j^*); P(\rmZ_j^*))  
    \!\le \!I(P_{mix}(\rmZ^*); P(\rmZ^*)).
\end{align}
Eq.~(\ref{eq:align}) shows that individually mapping each modality $\rmZ_j^*$ to $P_{mix}$ is a lower bound of mapping all modalities together to $P_{mix}$.
\end{proof}
Proposition~\ref{prop:P_mix} is \textbf{single-letterization} that simplifies the optimization problem over a large-dimensional (i.e., multi-letter) problem. 
Therefore, we individually align the representations of each modality to the anchor $P_{mix}$, rather than the whole distribution of all representations from all modalities:
\begin{equation}
\begin{split}
\label{eq:why_pmix}
    &\sum_{j=1}^{J} 
    \mathbb{E}_{\rmZ_j^*\sim P(\rmZ_j^*)}
    [\ln P(\rmY \! \mid \!\mathcal{C}^*(\rmZ_j^*)) -  \KL(P( \rmZ_j^*) \Vert P_{mix}) ]  \\
    \leq & \mathbb{E}_{\rmZ^*\sim P(\rmZ^*)}
    [\ln P(\rmY \! \mid \!\mathcal{C}^*(\rmZ^*)) \! \! - \!\! \KL(P( \rmZ^*) \Vert P_{mix}) ]. 
\end{split} 
\end{equation}

The former is termed Evidence Lower Bound (ELBO)~\cite{tomczak2018vae}, which is tighter than the latter. Thus,
minimizing the gap between all modalities and $P_{mix}$, the alternative objective for \emph{teacher} is further derived as:
\begin{equation}
 \label{eq:t_obj}
   \min
   \sum_{j=1}^J [\KL(P( \rmZ_j^*) \Vert P_{mix}) \!
    + \! H_c( P( \mathcal{C}^*(\rmZ_j^*)); P(\rmY))]. 
\end{equation}
As shown in Eq.~(\ref{eq:t_obj}),
the essential point is to find a feasible $P_{mix}$ that anchors all latent features in the space while {preserving the prediction ability from the latent features to targets for all \emph{students}}. 

\textbf{Possible approximations of $P_{mix}$.}
\label{approximations}
It is intractable to obtain the ideal $P_{mix}$ in practice. 
Therefore, different pre-defined $P_{mix}$ approximations are made.   
Similar to VAE, a possible assumption is that $P_{mix}$ is a fixed distribution such as standard normal distribution: $P_{mix}^{N}\triangleq \mathcal{N}(0, 1)$. 
However, $P_{mix}^{N}$ may not certify Proposition~\ref{prop:P_mix}, may yielding sub-optimal results.
Thus, we propose $P_{mix}$ which is one of $\{P(\rmZ^*_j)\}_{j=1}^J$ (i.e., $P_{mix}^{k} \triangleq P(\rmZ^*_{j=k})$ where $k\in \{1,...,J\}$) or which is a weighted mixture of them (i.e., $P_{mix}^* \triangleq \sum\nolimits_{j=1}^J w_j P(\rmZ_j)$) where $w_j$ is the associated weight of each modality.  $P_{mix}^{k}$ and $P_{mix}^*$ naturally preserve Proposition~\ref{prop:P_mix} through the prediction.  
Testing the effectiveness of the different $P_{mix}$, we find that a tractable approximation $P_{mix}^*$ will produce approximately the best results.

\section{Methodology}
\label{methodology}



\subsection{Proposed alignment paradigm}

The overall alignment paradigm consists of training \emph{teachers} and \emph{students}. 
The structure diagram of the \emph{teacher} is exhibited in the bottom left part of
Figure~\ref{fig:dist_vis}. 
For training the \emph{teacher}, each modality is initially encoded into the same latent space and then individually aligned to the pre-defined anchor $P_{mix}$ as shown in Eq.~(\ref{eq:align}). 
Finally, following the previous baselines,
a 3D U-Net is used as the predictor~($\mathcal{C^*}$) for segmentation based on the aligned latent features.

Then the enhanced prior knowledge obtained by the \emph{teachers} are leveraged to students by implanting to different backbones (see Sec.~\ref{Comparison}).
We train \emph{students} by distilling knowledge from the trained \emph{teachers} in the missing modality scenario
~\cite{azad2022smu,chen2021learning,hu2020knowledge,wang2023prototype}. The loss to be optimized is:
$
    L^S=  L_{seg}^\mathcal{
    S} +  L^{\mathcal{
    T}},
$
where $L_{seg}^\mathcal{S}$ represents the segmentation loss guided by ground truth labels, and $L^\mathcal{T}$ denotes the loss that receives supervision from the \emph{teacher} used in the previous works.


\subsection{Alignment with various $P_{mix}$}
\label{sec:align_details}
This alignment towards $P_{mix}$ standardizes data distributions from diverse sources into a consistent distribution, facilitating the learning of unique features across modalities.
Specifically, we provide details of the alignment to various empirical forms of $P_{mix}$ (which are denoted as $P^{k}_{mix}$, $P^{*}_{mix}$, and $P^{N}_{mix}$).

\textbf{Aligning to $P^{k}_{mix}$.}
\label{sec:pmix_align}
The $k^{th}$ modality that has the most feature invariant representation is a reasonable choice for $P_{mix}$ (see more details in Sec.~\ref{preperanchor}), i.e., $P_{mix}^{k} \triangleq P(\rmZ^*_{j=k})$. 
The alignment of other modalities to the chosen optimal modality is facilitated through Mean Squared Error (MSE). 
Here we minimize:
$
    \mathbb{E}[||\rmZ^*_j - \rmZ^*_k||^2]
$
for each paired sample. 

\textbf{Aligning to $P^{*}_{mix}$.}
In the quest to derive a more conducive latent space for integrating all modalities, we have advanced an innovative methodology termed Adaptive Alignment. This approach will transcend the basic alignment method that confines the latent space to a specific modality.
Adaptive Alignment operates under the presumption that an optimal latent space for a prior modality can serve as a foundational anchor. 
Then we have:
$
       \sum\nolimits_{j=1}^J w_j ||\rmZ^*_j - \rmZ^*_k||^2,
$
where $w_j$ are learnable weights.
Note that the \emph{teacher} is not frozen during training, with the purpose to enable the \emph{teacher} find the adaptive latent space.

\begin{table*}[t]
\scriptsize
\centering
\caption{Comparison of segmentation results in each class and average dice scores with different anchor $P_{mix}$s. For each modality, results are compared to a model that only uses unimodal and the original backbone models. Imp.: Improvement of $P^*_{mix}$ compare to original method. Average Imp. is the average improvement of all the Imp. column. The best average results of each backbone if all settings are highlighted.}
\label{tab:enhancedBest}
\resizebox{\linewidth}{!}{%
\begin{tabular}{c|cccc|cccc|cccc|cccc|cccc|c|c}
\toprule
\multicolumn{1}{c|}{\multirow{2}{*}{Method}} & \multicolumn{4}{c|}{Original Method} & \multicolumn{4}{c|}{without Alignment} & \multicolumn{4}{c|}{ with $P^N_{mix}$} & \multicolumn{4}{c|}{ with $P^k_{mix}$}         & \multicolumn{5}{c|}{ with $P^*_{mix}$ \textbf{Average Imp.: 1.75}} & \multicolumn{1}{c}{{}}\\
\cmidrule{2-23} 
\multicolumn{1}{c|}{}                        & WT      & TC      & EC     &  Avg.  & WT      & TC      & EC     & Avg.    & WT    & TC    & EC    & Avg.     &WT& TC    & EC    & Avg.          & WT     & TC     & EC    & Avg. &\textbf{Imp.} &Modality        \\
\midrule
\midrule
\rowcolor{mylightgray}Unimodal                                    & 72.96   & 65.59   & 37.77  & 58.77  &       &       &       &       &       &        &        &       &       &  & & & & & & &   \\
KD-Net                                         & {79.62}&59.83&33.69&57.72& 72.07	&66.22	&40.13	&59.47 &74.21&	{67.63}	&43.24  &\cellcolor{mygray}\textbf{61.69}  &74.06	&64.21	&41.78	&60.02  &71.49	&65.18	&{43.25}	&59.97  &{{2.25}} & \multirow{2}{*}{T1}  \\
PMKL                                        &73.31 &64.26 &41.37 &58.98&  {75.50}  & 65.98   & 40.09  &60.53   &{75.60}&	65.59	&43.31&	61.50  &75.06	&66.80	&41.43	&61.10  &72.04	&{68.39}	&{47.66}	&\cellcolor{mygray}\textbf{62.70} &{{3.97}}   \\
ProtoKD                                      & 74.46& {67.34}& {{47.41}}& {{63.07}}& 73.64   & 65.05   & {43.04}  & 60.57  &72.95&	65.52	&42.92&	60.47  &{75.60}	&66.95	&43.18	&\cellcolor{mygray}\textbf{61.91}  &73.98	&{67.36}	&42.11	&61.15 & -2.55   \\
SMU-Net & 74.33 & 65.52 & 40.22 & 60.02 & 75.24 & 68.52 & 43.03 & 62.26 & 75.10 & 66.41 & 42.78 & 61.43 & 75.15 & 67.25 & 41.71 & 61.37 & 75.02 & 67.78 & 43.30 &\cellcolor{mygray}\textbf{62.03}&2.01 \\
\midrule
\midrule
\rowcolor{mylightgray}Unimodal                                    & 82.65   & 66.76   & 45.23  & 64.91  &       &       &       &       &       &        &        &       &       &      & & & & & & &\\
KD-Net                                        & {{85.74}} &66.79 &33.63& 62.05&80.50&	66.99&	{48.02}&	65.17  &83.23	&69.64	&43.18	&65.35  &83.22	&70.72	&44.72	&66.22  &84.26&{71.30}&47.04&\cellcolor{mygray}\textbf{67.53}  &  {{5.48}}& \multirow{2}{*}{T2}  \\
PMKL                                      &81.00 &67.92& 47.09 &65.34   &  82.68  &  67.14  &  {44.82} & 64.88  &80.46&	69.06&	48.38&	65.97  &82.47	&69.56	&45.78	&65.94  &{83.77}	&{69.91}	&45.17	&\cellcolor{mygray}\textbf{66.28}   &{{0.94}}   \\
ProtoKD                                     &81.83& {68.29} &{47.35} &{65.82}  &   81.82 &  {70.21}  & 48.78  & \cellcolor{mygray}\textbf{66.94}  &83.82&	69.54&	45.03&	66.13  &{83.18}	&67.96	&{47.71}	&66.28  &83.01	&70.26	&47.29	&{66.85}   & {{1.03}}   \\
SMU-Net & 85.57 & 70.61 & 47.33 & 67.84 & 84.69 & 70.34 & 46.94 & 67.32 & 84.88 & 69.96 & 45.08 & 66.64 & 85.09 & 69.50 & 44.85 & 66.48 & 84.45 & 69.82 & 47.09 & \cellcolor{mygray}\textbf{67.12} & -0.62 \\
\midrule
\midrule
\rowcolor{mylightgray}Unimodal                                    & 71.41   & 73.30   & 76.36  & 73.69  &       &       &       &       &       &        &        &       &       &      & & & & & & &\\
KD-Net                                          & {78.87}&  80.83&  70.52&  76.74& 72.14   & 80.75  & 77.61  & 76.83  &72.49&	79.30&	74.46 &75.42  &76.73	& 81.64	& 75.56	& 77.98   &76.62	&80.15	&81.29	&\cellcolor{mygray}\textbf{79.36}&  {2.62}  & \multirow{2}{*}{T1ce}  \\
PMKL                                        &  70.50 & 76.92 & 75.54 & 74.32&   74.00 &  78.64  & 72.71  &  77.31 &73.89&	80.86&	77.48 &77.41 &77.46	& 80.71	& 75.40	& \cellcolor{mygray}\textbf{77.86}    &75.97	&80.35	&76.44	&77.58 &  {3.26} \\
ProtoKD                                      & 74.67 & {81.48} & {76.01} & {77.39} & 75.16   & 80.47   &  76.74 & 77.45 &76.52&	80.85&	75.73 &77.70  &75.98	& 79.41& 76.99	& 77.46    &74.91	&81.44	&77.39	&\cellcolor{mygray}\textbf{77.91}  & {0.52} \\
SMU-Net & 75.33 & 79.41 & 76.22 & 76.99 & 76.65 & 80.08 & 76.01 & 77.58 & 75.68 & 79.86 & 74.92 & 76.06 & 78.63 & 74.85 & 76.51 & 76.66 & 75.83 & 80.13 & 75.57 & \cellcolor{mygray}\textbf{77.18}&0.09\\
\midrule
\midrule
\rowcolor{mylightgray}Unimodal                                    & 81.91   & 63.57   & 40.74  & 62.07  &       &       &       &       &       &        &        &       &       &      & & & & & & &\\
KD-Net                                         & {88.28} &64.37 &33.39& 62.01&  84.97  &  63.16  & 41.44  & 63.19  &84.84 & 64.67 & 44.15 & 64.56 &  85.46	&66.77	&43.99	&\cellcolor{mygray}\textbf{65.41}  &84.96&	66.58&	42.16&	64.57  &   {2.56} & \multirow{2}{*}{Flair}  \\
PMKL                                       & 84.11& 62.21& 41.35 &62.56 & 84.74   &   67.07 & 43.42  &65.07   &84.09&	66.78&	42.13& 64.33   &83.84	&68.89	&41.41	&64.71    &85.70	& 68.44	& 43.57& 	\cellcolor{mygray}\textbf{65.90} &  {3.34} \\
ProtoKD                                  & 84.64 &{65.56}& {42.30}& {64.17}    & 84.59   & 67.70   & 40.91  &64.39   &   84.62  &   64.32     &  37.76      & 62.23         &84.23& 	67.73& 	41.45& 	64.47   &85.62	&68.71	&41.38	&\cellcolor{mygray}\textbf{65.23} &   {1.06}  \\
SMU-Net & 85.74 & 62.89 & 38.12 & 62.25 & 85.70 & 63.50 & 39.43 & 62.88 & 85.99 & 65.74 & 40.55 & 64.09 & 86.78 & 63.83 & 40.82 & 63.81 & 86.89 & 64.88 & 41.41 & \cellcolor{mygray}\textbf{64.39} & 2.14\\
\bottomrule
\end{tabular}%
}
\end{table*}

\begin{table*}[h]
\tiny
\centering
\caption{Average dice scores when aligning to modality T2 which is not ideal.}
\resizebox{\linewidth}{!}{%
\begin{tabular}{c|ccccc|ccccc|ccccc|ccccc}
\toprule
\multicolumn{1}{c|}{Method} &
\multicolumn{5}{c|}{T1} &
\multicolumn{5}{c|}{T2}    & 
\multicolumn{5}{c|}{T1ce} &
\multicolumn{5}{c}{Flair}
\\
\cline{2-21}
\multicolumn{1}{c|}{}                        & WT      & TC      & EC     & Avg. & \makecell[c]{T1ce \\ Avg.}    & WT    & TC    & EC    & Avg.  & \makecell[c]{T1ce \\ Avg.}  &  WT      & TC     & EC     & Avg. & \makecell[c]{T1ce \\ Avg.}     &  WT      & TC      & EC     & Avg. & \makecell[c]{T1ce \\ Avg.} \\
\midrule                                 
\midrule
\rowcolor{mylightgray}Unimodal                                    & 72.96   & 65.59   & 37.77  & 58.77  &  &82.65   & 66.76   & 45.23  & 64.91   &  &  71.41   & 73.30   & 76.36  & 73.69 &   & 81.91   & 63.57   & 40.74  & 62.07& \\
KD-Net                                          & 72.09 & 65.37 & 40.83 & 59.43 &\textbf{60.02} &80.78 & 67.64 & 45.59 & 64.67 &\textbf{66.22}&73.61 & 78.86 & 78.72 & 77.07 &\textbf{77.98}& 84.31 & 67.59 & 43.17 & 65.02  &\textbf{65.90}\\
PMKL                                        & 71.23 & 61.40 & 40.98 & 57.87 &\textbf{61.10}& 81.47 & 69.08 & 44.60 & 65.05 &\textbf{65.94}&71.92 & 76.22 & 77.71 & 75.29 &\textbf{77.86}& 84.00 & 65.18 & 42.53 & 63.90 &\textbf{64.71}\\
ProtoKD                                       & 72.95 & 67.37 & 40.73 & 60.35 & \textbf{61.91} &83.85 & 68.51 & 45.70 & 66.02&\textbf{66.28} &75.08 & 79.93 & 78.50 & \textbf{77.84} &77.46 & 84.36 &68.16 & 41.48 & 64.46 &\textbf{64.47}\\
SMU-Net&86.31&67.62&40.17&\textbf{64.70}&61.37&72.72&66.48&43.97&61.06&\textbf{66.48}&75.78&79.26&72.55&75.86&\textbf{76.66}&83.99&70.60&41.82&\textbf{65.47}&63.81\\
\bottomrule
\end{tabular}%
}
\label{tab:T2Result}
\end{table*}

\begin{table*}[h]
\tiny
\caption{Comparison of dice scores when different modalities are missing. $\bullet$ is the modality which is not missing, $\circ$ is the modality which is missing. $\Delta$ is the model with our paradigm.}
\centering

\resizebox{\linewidth}{!}{%
\begin{tabular}{c|c|cccc|cccccc|cccc|c|c}
\toprule
\multirow{4}{*}{Type} & Flair& $\circ$& $\circ$& $\circ$& $\bullet$& $\circ$& $\circ$& $\bullet$& $\circ$&$\bullet$&$\bullet$&$\bullet$&$\bullet$&$\bullet$& $\circ$& $\bullet$&\multirow{4}{*}{Avg.}\\
& T1& $\circ$& $\circ$& $\bullet$&$\circ$& $\circ$& $\bullet$& $\bullet$& $\bullet$&$\circ$&$\circ$&$\bullet$&$\bullet$&$\circ$& $\bullet$& $\bullet$&\\
& T1ce& $\circ$& $\bullet$& $\circ$&$\circ$& $\bullet$& $\bullet$& $\circ$& $\circ$&$\circ$&$\bullet$&$\bullet$&$\circ$&$\bullet$& $\bullet$& $\bullet$& \\
& T2 &$\bullet$& $\circ$& $\circ$&$\circ$& $\bullet$& $\circ$& $\circ$& $\bullet$&$\bullet$&$\circ$&$\circ$&$\bullet$&$\bullet$& $\bullet$& $\bullet$&\\
\bottomrule
\multirow{2}{*}{WT}& PMKL&81.00&70.50&73.31&84.11&75.82&66.62&79.85&79.35&83.01&75.67&73.86&84.78&83.19&70.57&85.62&77.82\\
& $\Delta$&83.77&75.97&72.04&85.70&77.98&74.34&83.82&83.00&85.09&83.85&81.86&85.27&86.13&82.15&86.81&\textbf{81.85}\\
\bottomrule
\multirow{2}{*}{TC}&PMKL&67.92&76.92&64.26&62.21&72.46&70.05&51.64&66.66&68.74&66.44&70.12&61.79&75.04&68.89&80.14&68.22\\
&$\Delta$&69.91&80.35&68.39&68.44&75.06&76.04&54.62&67.49&67.96&67.66&76.55&65.49&76.70&78.21&79.22&\textbf{71.47}\\
\bottomrule
\multirow{2}{*}{ET}&PMKL&47.09&75.54&41.37&41.35&70.25&69.31&21.92&47.10&44.13&62.43&68.24&39.90&69.28&65.37&75.01&55.89\\
&$\Delta$&45.17&76.44&47.66&43.57&73.37&77.88&36.78&46.00&45.24&61.90&78.79&44.86&75.47&77.96&77.85&\textbf{60.60}\\
\bottomrule
\end{tabular}%
}
\label{tab:otherMissing}
\end{table*}

\textbf{Aligning to $P^{N}_{mix}$.}
Since $P^k_{mix}$ is the selected form ${P(\rmZ^*_j)}_1^J$ and $P^*_{mix}$ is a weighted mixture of ${P(\rmZ^*_j)}_1^J$'s, 
the relationship between $\rmZ^*_j$ and these $P_{mix}$'s
is tractable, and the feature and its target for alignment are paired. Therefore, we can exploit MSE as the alignment loss. However, $P^N_{mix}$ utilizes the standard Gaussian distribution, not derived from a weighted combination of $\{P(\rmZ^*_j)\}_1^J$, its relationship with $P(Z^*_j)$ is unclear. As such, the definitive alignment target for each feature remains unknown. 
Consequently, we can only use the KL divergence to align the entire distribution of $P(Z^*_j)$ to $P^N_{mix}$.
Similar to VAE, 
given $J$ modalities, our objective is to minimize
$\mathbb{E}[\KL(P(\rmZ^*_j) \Vert P^{N}_{mix} )]$, 
which is equivalent to minimizing: 
$
    \mathbb{E}[2\log {1}/{v(\rmZ^*_j)} + {v(\rmZ^*_j)^2 + ({\Bar{\rmZ}^*_j})^2}/{2} - {1}/{2}],
$ which was proposed in VAE.
Here $v(\rmZ^*_j)$, $\Bar{\rmZ}^*_j$ denotes the variance and mean of ${\rmZ}^*_j$, which are obtained by learnable parameters during the training process.

\section{Experiments}
\label{sec:exp}
\textbf{Data and implementation details.}
The 2018 Brain Tumor Segmentation Challenge (BRATS) dataset ~\cite{bakas2017advancing,menze2014multimodal}, consisting of 285 subjects with four MRI modalities (T1, T1c, T2, and FLAIR), is employed to evaluate the proposed paradigm and other baselines. 
Annotations are given by normal tissue regions and three tumor-related masks, i.e., whole tumor (WT), tumor core (TC), and enhancing core (EC). Image intensities are normalized to $[-1,1]$. 
Each volume is augmented by randomly cropping each training example as $80\times80\times80$~\cite{chen2021learning}, while the dataset is split following~\cite{wang2023prototype}.
\emph{Teachers} and \emph{students} are optimized with Adam. Batch size is set as 4. Learning rates are initialized as $1e^{-3}$ which are gradually decayed by $1e^{-5}$ for both \emph{teachers} and \emph{students}.

\subsection{Comparison with state-of-the-art methods}
\label{Comparison}
Table~\ref{tab:enhancedBest} reports how our proposal could promote state-of-the-art (SOTA) approaches in guiding \emph{students} 
, including KD-Net~\cite{hu2020knowledge}, PMKL~\cite{chen2021learning}, ProtoKD~\cite{wang2023prototype} and SMU-Net~\cite{azad2022smu}. When three modalities (the most challenging setting) are missing,  the anchor $P^*_{mix}$, could lead to an improvement of 1.75 dice score on average for various SOTA \emph{students}. We also carry out experiments on other less difficult scenarios. These additional results can be referred to in Table~\ref{tab:otherMissing}.


\begin{table}[t]
\centering
\begin{minipage}{0.48\textwidth}
\centering
\tiny
 \caption{Comparison results of \emph{teacher} with different types of $P_{mix}$. Best results are highlighted in \textbf{bold}.}
\label{tab:teacherResult}
\begin{tabular}{ccccc}
\toprule
 $P_{mix}$           & WT    & TC    & EC    & Average \\
\midrule
Without $P_{mix}$                     &85.62  &\textbf{80.14}  & 75.01 & 80.26\\
$P^N_{mix}$                    &    84.14         & 77.44      & 74.82      &  79.13\\
$P^k_{mix}$                         &  86.34   & 79.75 & 76.91 & 81.00 \\
\rowcolor{mygray} $P^*_{mix}$              &  \textbf{86.81}& 79.22 & \textbf{77.85} & \textbf{81.29} \\                  
\midrule
\end{tabular}
\end{minipage}
\hfill 
\begin{minipage}{0.48\textwidth}
\centering
\tiny
\caption{Comparison results of \emph{teachers} on different single- modality. Columns represent target modalities. }
\label{tab:domainSelect}
\begin{tabular}{c|cccc|c}
\toprule
Modality & T1          & T2          & T1ce        & Flair       & Average   \\ \midrule
T1     & 58.92       & { 55.27} & 9.89        & 22.15       & 36.56  \\ 
T2     & 4.29        & 65.01       & 11.87       & 34.20       & 28.84  \\ 
T1ce   & { 38.19} & 11.70       & 76.76       & { 37.95} & \textbf{41.15}  \\ 
Flair  & 5.05        & 40.47       & { 40.17} & 62.14       & 36.96  \\ 
\bottomrule
\end{tabular}
\end{minipage}
\end{table}

\subsection{Find the \emph{teacher} with the best prior knowledge }
\label{findbestteacher}

Figure~\ref{fig:dist_vis} presents distributions of anchors among $P^N_{mix}$, $P^k_{mix}$, and $P^*_{mix}$. 
Consistent with the theoretical analysis presented in Sec.~\ref{approximations}, modality gaps are not narrowed at all.  As such, the space cannot be filled with $P^N_{mix}$ in (c), suggesting that $P_{mix}^{N}\triangleq \mathcal{N}(0, 1)$ cannot be placed as a fixed anchor to be aligned to.
On the contrary, as shown in (d) and (e) when anchors $P^k_{mix}$ and $P^*_{mix}$ are employed, distributed centers of each modality are almost overlapped, demonstrating that modality gaps are narrowed. 
Table~\ref{tab:teacherResult} statistically demonstrates that $P^*_{mix}$, as the best anchor for \emph{teacher},  generates an improvement by $1.03$ dice score. 
\subsection{Effectiveness verification of the \emph{student}}

\textbf{Proper $P^k_{mix}$ to be aligned.}
\label{preperanchor}
To design a pre-defined anchor $P^k_{mix}$, we train the \emph{teachers} with each single modality. As shown in Table~\ref{tab:domainSelect}, 
treating T1ce modality as $P^k_{mix}$ achieves the best dice score $41.15$ across all the testing sets.
This implies that T1ce encompasses the most comprehensive information, making it an ideal fixed anchor for enhancing feature invariant representation learning. 
Additionally, we also find that the learning parameters in $P^*_{mix}$ of T1, T2, T1ce and Flair are $0.7759$, $0.8977$, $2.2055$ and $0.3055$ respectively.
Apparently, T1ce enjoys the biggest weight, meaning that it is the most informative modality. 
Conversely, an unsuitable fixed anchor will significantly impair segmentation performance, as discussed in Table~\ref{tab:T2Result}.

\textbf{Options of anchors.}
Table~\ref{tab:enhancedBest} also shows that \emph{students} trained by \emph{teachers} with $P^*_{mix}$ performs better than \emph{teachers} with $P^k_{mix}$; in addition, \emph{teachers} with $P^N_{mix}$ is the worst one.
$P^*_{mix}$ brings an improvement on \emph{students} by 1.75 dice score on average while $P^k_{mix}$ brings 1.28.
Visualized demonstrations of aligning the latent representations to different anchors are displayed in Figure~\ref{fig:seg_visual}.
In Figure~\ref{fig:enter-label}, visualization results of the latent space for a specific sample are also shown.
As observed,  \emph{teachers} with $P^N_{mix}$ and $P^*_{mix}$ are more sensitive to small targets (yellow boxes). 
Moreover, \emph{teachers} with $P^k_{mix}$ and $P^*_{mix}$ are more likely to produce less false detection (white boxes). 
Overall, we can conclude that $P^*_{mix}$ is indeed one effective anchor.


\begin{figure}[t]
  \centering
  \includegraphics[width=\linewidth]{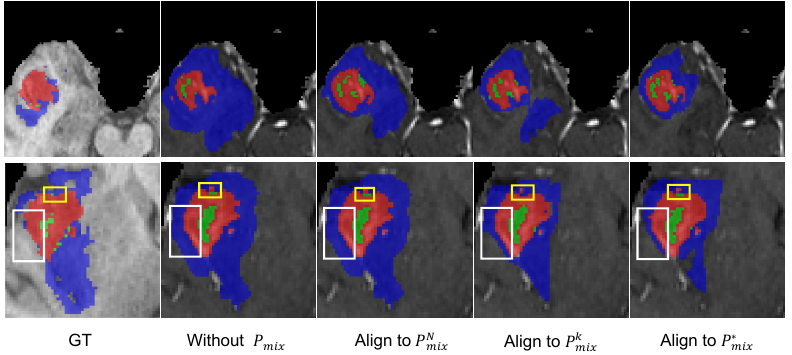}
   \caption{Segmentation visualization on~BraTS2018~\cite{bakas2017advancing,menze2014multimodal}.}
   \label{fig:seg_visual}
\end{figure}

\section{Conclusion}
In this paper, we present a novel alignment framework to narrow the modality gaps whilst learning simultaneously invariant feature representations in segmenting brain tumors with missing modalities. 
Specifically, we invent an alignment paradigm for the \emph{teacher} with latent space distribution $P^*_{mix}$ as the aligning anchor, thereby building the reliable prior knowledge to supervise training \emph{students}.
Meanwhile, we provide theoretical support for the proposed alignment paradigm, demonstrating that individually aligning each modality to $P_{mix}$ certifies a tighter evidence lower bound than mapping all modalities as a whole.
Extensive experiments have demonstrated the superiority of the proposed paradigm in several latest state-of-the-art approaches, enabling them to better transfer knowledge from the multi-modality \emph{teacher} to the \emph{student} with missing modalities.

\begin{figure}[]
    \centering
    \includegraphics[width=\linewidth]{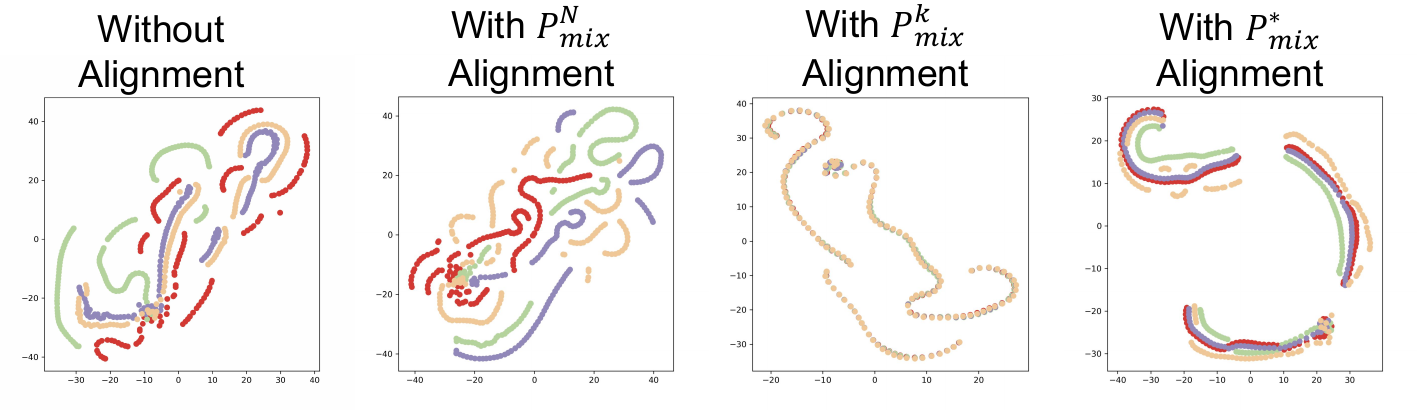}
    \caption{Latent space feature visualization of \textit{teacher}s with different anchor $P_{mix}$s for a specific sample. Different colors represent different modalities. Latent spaces aligned to $P^N_{mix}$ and without alignment have obvious modality gaps, while latent spaces aligned to $P^k_{mix}$ and $P^*_{mix}$ have narrow gaps.}
    \label{fig:enter-label}
\end{figure}



\clearpage
\bibliographystyle{IEEEtran}
\bibliography{ref}

\end{document}